\documentclass[11pt,a4paper]{article}

\title{Wedge-Local Quantum Fields on a Nonconstant Noncommutative Spacetime}
\author{A. Much\\ \footnotesize{Max-Planck-Institute for Mathematics in 
the Sciences, 04103 Leipzig, Germany}
\\ \footnotesize{Institute for Theoretical Physics, University of 
Leipzig, 04009 Leipzig, Germany}}
\usepackage{geometry}
\usepackage{mathrsfs}
\usepackage{stmaryrd}
\usepackage{hyperref}
\usepackage{amsmath}
\usepackage{amsthm}
\usepackage{amsfonts}
\usepackage{textcomp}
\usepackage{url}
\usepackage[ansinew]{inputenc}
\newtheorem{theorem}{Theorem}[section]
\newtheorem{lemma}[theorem]{Lemma}
\newtheorem{proposition}[theorem]{Proposition}

\newtheorem{definition}[theorem]{Definition}
 
\usepackage{fancyhdr}
\begin{document}

\maketitle
\abstract{Within the framework of warped convolutions we deform the 
massless free scalar field. The deformation is performed by using the 
generators of the special conformal transformations. The investigation 
shows that the deformed field turns out to be wedge-local. Furthermore, 
it is shown that the spacetime induced by the deformation with the 
special conformal operators is nonconstant noncommutative. The 
noncommutativity is obtained by calculating the deformed commutator of 
the coordinates.
  }\newpage
\tableofcontents
\section{Introduction}
Noncommutative quantum field theories (NCQFT) enjoy wide popularity 
among theoretical physicists. From a string theoretical point of view, 
NCQFT became popular due to the observation that it can be obtained in a 
certain low energy limit from string theory \cite{SWS}. From a quantum 
field theoretical aspect, NCQFT gained interest due to many reasons. 
Most importantly it was thought that by the introduction of a 
fundamental length renormalisation ambiguities would disappear and 
ultra-violet (UV) divergences would be canceled. But already in first 
order of perturbation theory, the euclidean noncommutative $\phi^4$ 
model exhibited a new type of divergences. The new divergences could not 
be cured by standard renormalisation procedures. In a series of papers 
\cite{GW1, GW2, GW3} the authors added a term to the noncommutative  
$\phi^4$ model, based on duality considerations, and proved the 
renormalisability to all orders in perturbation.
\newline\newline
Quantum field theory on a noncommutative Minkowski spacetime was 
rigorously realised in \cite{DFR}. The quantum field therein was 
defined on a tensor product space $\mathcal{V}\otimes\mathscr{H}$. Where $\mathscr{H}$ is the Bosonic Fock space and $\mathcal{V}$ is the representation space of the noncommuting coordinate operators $\hat{x}_{\nu}$, satisfying the Moyal-Weyl plane commutator relations $[\hat{x}_{\mu},\hat{x}_{\nu}]=i\theta_{\mu\nu}$. The matrix $ \theta_{\mu\nu}$ is a constant, nondegenerate and skew-symmetric matrix.
Many authors \cite{{AB}, {G}, {GL1}} succeeded representing the scalar field on $\mathscr{H}$ instead of 
$\mathcal{V}\otimes\mathscr{H}$. Furthermore, in \cite{GL1} this representation was used to construct a map  from the set of skew-symmetric matrices, which describe the noncommutativity, 
to a set of wedges. In the next step the construction was applied to map
the noncommutative scalar field to a scalar field living on a wedge. The respective model led to 
weakened locality and covariance properties of the field and to a 
nontrivial S-matrix. The result is astonishing because notions of 
covariance and locality are usually lost on a noncommutative spacetime.
\newline\newline
The method of deformation was further generalised in 
\cite{{BS},{BLS},{GL4}} and was made public
under the name of warped convolutions. It is interesting to note that 
the model formulated in
\cite{GL1} can be obtained from warped convolutions by using the 
momentum operator $P_{\mu}$ for the
deformation.
The method in \cite{GL4} was also successfully used to define deformations of 
a scalar massive Fermion,
\cite{A}.
\newline\newline
In fact any strongly continuous unitary representation of the group 
$\mathbb{R}^n$
can be used to deform the free scalar field. In the following work we 
deform the QFT of a free
scalar field with the special conformal operators using the method of 
warped convolutions. We show
that the resulting noncommutative spacetime is nonconstant. Furthermore, 
we show that the
constructed model exhibits covariance and locality properties which are 
highly nontrivial for a
nonconstant noncommutative spacetime. 
\newline\newline
The organisation of the paper is as follows: 
In Sec. 2, we give a brief introduction of the conformal group and the 
isomorphism to the pseudo-orthogonal group $SO(2,d)$. The proof of  
self-adjointness of the special conformal operators was given rigorously 
in \cite{SV} and is sketched in Sec. 3. The proof therein relies on 
the fact that the momentum operator and the special conformal operator 
are unitarily equivalent. Furthermore, we deform the free scalar field 
with the special conformal operators and use the unitary equivalence to 
proof convergence of the deformation in the Hilbert space norm. The 
Wightman properties, transformation properties and wedge-locality of the 
deformed field are proven in Sec. 4.  In Sec. 5, we show how the 
deformation with the special conformal operators leads to a nonconstant 
noncommutative spacetime. This is done by calculating the commutator of 
the coordinates using the deformed product given in \cite{BLS}.
\section{The conformal group and SO(2,d)}
\subsection{Generators of the conformal group}
A conformal transformation  of the coordinates is defined to be an 
invertible mapping $\mathbf{x}'\rightarrow \mathbf{x}$, which leaves the 
$d$-dimensional metric $g $  invariant up to a scale factor, \cite{DMS}:
\begin{equation}  \label{cond1}
g'_{\mu\nu}(x')=F(x)g_{\mu\nu}(x).
\end{equation}
The mappings satisfying the condition (\ref{cond1}) are the Lorentz 
transformations, the translations, the dilations and the special 
conformal transformations. These transformations are generated by the 
operators $L_{\mu\nu}$, $P_{\rho}$, $D$, $K_{\sigma}$ and the set of all 
conformal transformations forms the conformal group.
\\\\
  The conformal algebra is defined by the commutation relations of the 
generators and is given as follows:
\begin{equation}
[L_{\mu\nu},L_{\rho\sigma}]
=i\left(\eta_{\mu\sigma}L_{\nu\rho}+\eta_{\nu\rho}L_{\mu\sigma}-\eta_{\mu\rho}L_{\nu\sigma}-
\eta_{\nu\sigma}L_{\mu\rho},
\right)
\end{equation}
\begin{equation}
[P_{\rho},L_{\mu\nu}]=i\left(\eta_{\rho\mu}P_{\nu}-\eta_{\rho\nu}P_{\mu}\right),
\qquad 
[K_{\rho},L_{\mu\nu}]=i\left(\eta_{\rho\mu}K_{\nu}-\eta_{\rho\nu}K_{\mu}\right),
\end{equation}
\begin{equation}
[P_{\rho},D]=iP_{\rho}, \qquad [K_{\rho},D]=-iK_{\rho},
\end{equation}
\begin{equation}
[P_{\rho},K_{\mu}]= 2i\left(\eta_{\rho\mu}D-L_{\rho\mu}\right),
\end{equation}
   with all other commutators being equal to $0$.

\subsection{Isomorphism between the conformal group and $SO(2,d)$}
To see the isomorphism between the conformal group in d dimensions and 
the pseudo-orthogonal group  $SO(2,d)$,  one introduces the following 
definitions:

\begin{equation}
J_{4,\mu}:=\frac{1}{2}\left(P_{\mu}-K_{\mu}\right),\qquad 
J_{5,\mu}:=\frac{1}{2}\left(P_{\mu}+K_{\mu}\right),
\end{equation}
\begin{equation}
J_{\mu}^{\pm}:=J_{5,\mu}\pm J_{4,\mu}
  \qquad J_{-1,0}:=D,\qquad J_{\mu\nu}:=L_{\mu\nu},
\end{equation}
\begin{equation}
J_{ab}=-J_{ba},\qquad a,b=0,1,\dots,d,d+1.
\end{equation}
The defined generators $J_{ab}$ obey the algebra of $SO(2,d)$ with the 
following commutator relations:
\begin{equation}
[J_{ab},J_{cd}]
=i\left(\eta_{ad}J_{bc}+\eta_{bc}J_{ad}-\eta_{ac}J_{bd}-
\eta_{bd}J_{ac}
\right),
\end{equation}
where the diagonal metric has the following form
\begin{equation}
\eta_{aa}=(+1,\underbrace{-1,..,-1}_{d},+1).
\end{equation}
This shows the isomorphism between the conformal group and $SO(2,d)$. As 
one can easily see, the full conformal group contains the Poincar\'{e} 
group as a subgroup.
\section{Deforming the scalar quantum field}
The method used for deformation in this work was introduced in 
\cite{BS,BLS} and goes by the name of warped convolutions. This device 
can be used to deform relativistic quantum field theories and the hope 
is in near future to solve a nontrivial interacting QF model in 4 
dimensions. There were interesting results in \cite{GL2} where the authors obtained a non trivial S-matrix differing from the free one by momentum dependent phases.  The results introduced in \cite{GL1, GL2}  
can be obtained by using the framework of warped convolutions to deform 
the underlying QFT. This is done by using the momentum operator as the 
generator of a strongly continuous unitary representation of the group 
$\mathbb{R}^n$. \newline\newline
Our approach in this work will be to use the special conformal operators 
as the generators of such representations. This will lead to a new QFT 
model which on one hand can be interpreted as a QF on a nonconstant 
noncommutative spacetime, and on the other hand as a wedge-local QFT model.

\subsection{Self-adjointness of the special conformal operators}
To proceed with deforming via warped convolutions, it is necessary to 
prove self-adjointness of the special conformal operator $K_{\mu}$.  The 
proof was given in \cite{SV} relying on the fact that the special 
conformal operator can be defined as
\begin{equation}\label{puk}
K_{\mu}:= U_{R}P_{\mu} U_{R},
\end{equation}
where $ U_{R}$ is the inversion operator. The reason for the definition 
is that any special conformal transformation
\begin{equation}
U(b)x^{\mu} =\frac{x_{\mu}-b_{\mu}x^2}{1-2b_{\mu}x^{\mu}+b^2x^2},
\end{equation}
can be written as a product
\begin{equation}
U(b)=U_{R}T(b)U_{R},
\end{equation}
of a translation $T(b)x^{\mu}=x^{\mu}+b^{\mu}$ and inversions 
$U_{R}x^{\mu}=-{x^{\mu}}/{x^2}$.
By constructing  a self-adjoint unitary representation $ U_{R}$ in 
$\mathscr{H}_{1}:=L^2(d^n\mu(\mathbf{p}),\mathbb{R}^n)=\{f : \int 
d^n\mu(\mathbf{p})  |f(\mathbf{p})|^2<\infty, d^n\mu(\mathbf{p}):= 
d^n\mathbf{p} (2|\mathbf{p}| )^{-1}
\}$ for $n\geq 1$,  the essential self-adjointness of the operators $K^{\mu}$  on the 
dense domain $\Delta (R):= U_{R}\Delta (P)$ follows. $\Delta (P)$ is the 
dense domain of all functions from $ \mathscr{H}_{1} $ vanishing at 
infinity faster than any inverse polynomial in $p^k$ and is given as follows
\begin{equation}
\Delta  (P) =\{f\in \mathscr{H}_{1} :
|\left(\mathbf{p}^2
\right)^rf(\mathbf{p})|\leq c_r(f)<\infty;\quad r=0,1,2,\dots\}.
\end{equation}
Due to the unitary equivalence (\ref{puk}),  $P_{\mu}$ and $K_{\mu}$ 
have the same spectrum contained in the closed forward cone
  \begin{equation}
  \overline{V^{+}_0}:=\{p^{\mu}: p^{\mu}p_{\mu}\geq 0, p_{0}\geq 0\}.
\end{equation}
The last step in \cite{SV} consists in showing that the special conformal 
operator defined in the following way (\ref{puk}), is identical with the 
special conformal generator of the conformal group.

\subsection{Special conformal transformation of the free scalar field}
Since in the context of the present paper we need the transformation of 
the free scalar field under the special conformal group, we shall also 
briefly summarise those results obtained in \cite{SV}.\newline\newline
For $n=1$ the existence of a unitary representation for the whole 
conformal group was proven. The special conformal operator transforms 
the free scalar field $\phi(x)$ in the following manner
\begin{equation}
\alpha_{b}(\phi(x)):=
e^{ib_{\mu}K^{\mu}} \phi(x)e^{-ib_{\mu}K^{\mu}} =
\phi(x_b)-\phi(-\frac{b}{b^{\mu}b_{\mu}}),
\end{equation}
where
\begin{equation}
  x^{\mu}_b:=\frac{x_{\mu}-b_{\mu}x^2}{1-2b_{\mu}x^{\mu}+b^2x^2}.
\end{equation}
In the two dimensional spacetime test functions $f\in 
\mathscr{S}(\mathbb{R}^2)$ which are used to
smear the distribution valued operator $\phi(x)$ are chosen to satisfy 
$\int d^2x f(x)=0$. The
reason for this specific choice is to circumvent IR-divergences and it 
will be used through the
entire work.\\\\
Now if $n=2l+1$ for $l\in\mathbb{N}$ one obtains the following result
\begin{equation}\label{trafo1}
\alpha_{b}(\phi(x))=\sigma_b(x)^{\frac{1-n}{2}}\phi(x_b),
\end{equation}
where
\begin{equation}\label{sf}
\sigma_b(x):= {1-2b_{\mu}x^{\mu}+b^2x^2}.
\end{equation}
It was further proven that one only obtains a unitary representation for 
the whole conformal group if  $n=4l+1$ for $l\in\mathbb{N}$. In the 
other cases for odd $n$ one has to deal with representations of the 
covering of the conformal group. The reason for the non-existence of a 
unitary representation for the whole conformal group lies in the non 
positivity of the scale factor $\sigma_b(x)$. For the present paper this 
will become important due to our intention to formulate the model in four
spacetime dimensions. In Sec. 4, we will prove that the scale factor 
$\sigma_b(x)$ is positive for a scalar field localised in the wedge. 
Therefore, we will not have a problem obtaining a unitary representation 
for the whole conformal group.
 
\subsection{Deforming the QF with special conformal operators}
In this section we deform the massless scalar field with the special 
conformal operators using the framework warped convolutions. To proceed 
with the deformation, we first define the undeformed free scalar field 
$\phi$ with mass $m=0$ on the $n+1$-dimensional Minkowski spacetime as 
an operator valued distribution acting on its domain in the Bosonic Fock 
space. Such a particle with momentum $\mathbf{p} \in \mathbb{R}^n$ has 
the energy defined by $\omega_{\mathbf{p} }=|\mathbf{p}|$.
\begin{definition} The Bosonic Fock space $\mathscr{H}^{+}$ is defined 
as in \cite{Fr,S}:
\begin{equation*}
\mathscr{H}^{+}=\bigoplus_{m=0}^{\infty}\mathscr{H}_{m}^{+}
\end{equation*}
where the m particle subspaces are given as
  \begin{align*}
\mathscr{H}_{m}^{+}&=\{\Psi_{m}: \partial V_{+} \times  \dots \times 
\partial V_{+} \rightarrow \mathbb{C}\quad \mathrm{symmetric}
|\\ &\left\Vert  \Psi_m \right\Vert^2 =\int 
d^n\mu(\mathbf{p_1})\dots\int d^n\mu(\mathbf{p_m})
|\Psi_{m}(\mathbf{p_1},\dots,\mathbf{p_m})|^2<\infty\},
\end{align*}
with
  \begin{equation*}
\partial V_{+}:=\{p\in \mathbb{R}^d|p^2=0,p_0>0\}.
\end{equation*}
  \end{definition}
The particle annihilation and creation operators can be defined by their 
action on m-particle wave functions
  \begin{align*}
(a(f)\Psi)_m(\mathbf{p_1},\dots,\mathbf{p_m})&=\sqrt{m+1}\int 
d^n\mu(\mathbf{p})\overline{f(\mathbf{p})}
\Psi_{m+1}(\mathbf{p},\mathbf{p_1},\dots,\mathbf{p_m})\\
(a(f)^{*}\Psi)_m( \mathbf{p },\mathbf{p_1},\dots,\mathbf{p_m})&= \left\{
\begin{array} {cc}
0, \qquad &m=0 \\ \frac{1}{\sqrt{m}}\sum\limits_{k=1}^{m} f(\mathbf{p_k})
\Psi_{m-1}(\mathbf{p_1},\dots,\mathbf{p_{k-1}},\mathbf{p_{k+1}},\dots,\mathbf{p_m}),\quad 
&m>0
\end{array} \right.
\end{align*}
with $f \in \mathscr{H}_{1} $ and $\Psi_m \in \mathscr{H}_{m}^{+}$ . The commutator relations of $a(f), 
a(f)^{*}$ follow immediately and are given as follows
  \begin{align*}
[a(f), a(g)^{*}]=(f,g)=\int d^n\mu(\mathbf{p}) \overline{f(\mathbf{p})} 
g(\mathbf{p}), \qquad
[a(f), a(g)]=0=[a^{*}(f), a^{*}(g)].
\end{align*}
Particle annihilation and creation operators with sharp momentum are 
introduced as operator valued distributions and are given as follows
  \begin{align*}
a(f)=\int d^n\mu(\mathbf{p}) \overline{f(\mathbf{p})}a(\mathbf{p}), 
\qquad a(f)^{*}=\int d^n\mu(\mathbf{p}) {f(\mathbf{p})}a^{*}(\mathbf{p}),
\end{align*}
where the particle annihilation and creation operators with sharp 
momentum satisfy the following commutator relations
  \begin{align*}
[a(\mathbf{p}), a(\mathbf{q})^{*}]=2\omega_{\mathbf{p} 
}\delta^n(\mathbf{p}-\mathbf{q}), \qquad
[a(\mathbf{p}), a(\mathbf{q})]=0=[a^{*}(\mathbf{p}), a^{*}(\mathbf{q})].
\end{align*}
In the next step we define the warped convolutions of the free scalar 
field. This is done using the essential self-adjointness of the 
generators $K_{\mu}$ which in turn define a unitary operator $U(b):= 
e^{ib_{\mu}K^{\mu}}$.  The definition of the operator valued function 
$U(b)$ leads to a strongly continuous unitary representation of 
$\mathbb{R}^d$, for each $b_{\mu}\in \mathbb{R}^{d}$. This can be proven 
by making use of Stone's theorem, \cite{RS}.  To define the deformation, 
we need the the unitary operator of translations defined by $T(y):= 
e^{iy_{\mu}P^{\mu}}$, for each $y_{\mu}\in \mathbb{R}^{d}$, and the
extended dense domain  $\Delta_m(P):=\bigotimes_{k=1}^{m} 
\Delta(P) $. Furthermore, we define
$\Gamma( U_{R}):=\bigotimes_{k=1}^{m} U_{R} $ to be the unitary operator 
of the inversions on $ \mathscr{H}^{+}_{m}$, \cite{RS}. From the former 
definitions the extended domain $\Delta_m (R) =\Gamma( U_{R})\Delta_m 
(P) $ follows.
\begin{definition}  \label{deff}
  Let $\theta$ be a real skew-symmetric matrix w.r.t. the Minkowski 
scalar-product on $\mathbb{R}^{d}$ and let $\phi(f)$ be the free scalar 
field smeared out with functions $f \in \mathscr{S}(\mathbb{R}^d)$. Then 
the operator valued distribution $\phi(f)$ deformed with the special conformal operators, 
denoted as $\phi_{\theta,K}(f)$,  is defined on vectors of the dense domain 
$\Delta_m (R)$  as follows
\begin{align}\label{defe}
\phi_{\theta,K}(f)\Psi_{m}:&=(2\pi)^{-d}
  \iint  d^{d}y d^{d}k e^{-iy_{\mu}k^{\mu}}\alpha_{\theta 
y}(\phi(f))U(k)\Psi_{m}\\&=
(2\pi)^{-d}
  \iint  d^{d}y d^{d}k e^{-iy_{\mu}k^{\mu}}\alpha_{\theta 
y}\left(a(\overline{f^-})+a^*(f^+)\right)U(k)\Psi_{m}\\&
=:\left(a_{\theta,K}(\overline{f^-})+a_{\theta,K}^*(f^+)\right)\Psi_{m}.
\end{align}
The test functions $f^{\pm}(\mathbf{p})$ in momentum space are defined 
as follows
\begin{equation*}
f^{\pm}(\mathbf{p}):=\int d^{d}xf(x)e^{\pm ipx}, \qquad 
p=(\omega_{\mathbf{p}},\mathbf{p})\in \partial V_{+}.
\end{equation*}
\end{definition}
By arguing with the essential self-adjointness of the special conformal operator, it can be shown that the integral (\ref{defe}) converges in the strong operator topology if the undeformed operator is bounded, \cite{BLS}. Due 
to the fact that we are dealing with an unbounded operator it is not 
clear in what sense the deformation with the special conformal operator 
converges. To show that the warped convolutions (\ref{defe}) exist in Hilbert space norm, we 
use the unitary equivalence (\ref{puk}) between the momentum operator 
and the special conformal operator.

\begin{lemma} \label{ldk1}
  For $f \in \mathscr{S}(\mathbb{R}^d)$ and $\Psi_m \in \Delta_m(R)$, a 
transformation exists that maps the field deformed with the momentum 
operator $\phi_{\theta,P}(f)$ to the field  deformed with the special 
conformal operator $\phi_{\theta,K}(f)$. This transformation is given as 
follows
\begin{equation*}
  \phi_{\theta,K}(f)\Psi_m=
\Gamma( U_{R})\left(\Gamma( U_{R})\phi(f)\Gamma( 
U_{R})\right)_{\theta,P}\Gamma( U_{R})\Psi_m.
\end{equation*}
\end{lemma}
\begin{proof} 
By using the unitary equivalence (\ref{puk}) the lemma is easily proven
\begin{align*}
  \phi_{\theta,K}(f)\Psi_m&=(2\pi)^{-d}
  \iint  d^{d}y d^{d}k e^{-iy_{\mu}k^{\mu}}U(\theta y)\phi(f)U(-\theta 
y+k)\Psi_{m}\\
&= (2\pi)^{-d}\iint  d^{d}y d^{d}k e^{-iy_{\mu}k^{\mu}}\Gamma( 
U_{R})T(\theta y)\Gamma( U_{R}) \phi(f)
\Gamma( U_{R})T(-\theta y+k)\Gamma( U_{R}) \Psi_{m}
\\
&=
\Gamma( U_{R})\left(\Gamma( U_{R})\phi(f)\Gamma( 
U_{R})\right)_{\theta,P}\Gamma( U_{R})\Psi_m. 
\end{align*}
\end{proof}

\begin{lemma} \label{ldk2}
  For  $\Phi_m \in \Delta_m(R)\subset \mathscr{H}_{m}^{+}$ the familiar 
bounds of  the free field hold for the deformed field 
$\phi_{\theta,K}(f)$ and therefore the deformation with the special 
conformal operators exists in the Hilbert space norm.
\end{lemma}
\begin{proof} 
By using lemma \ref{ldk1} one obtains the familiar bounds for a free 
scalar field. For $\Phi_m \in
\Delta_m(R)$ there exists a $\Psi_m \in \Delta_m(P)$ such that the 
following holds
\begin{align*}&\left\Vert\phi_{\theta,K}(f)\Phi_m\right\Vert=
\left\Vert\phi_{\theta,K}(f)\Gamma( U_{R})\Psi_m\right\Vert= \left\Vert 
(\Gamma( U_{R})\phi(f) \Gamma( U_{R}))_{\theta,P}\Psi_m \right\Vert 
=\left\Vert (\phi( U_{R}f))_{\theta,P}\Psi_m \right\Vert \\&\leq 
\left\Vert ( a( U_{R}\overline{f^-}))_{\theta,P}\Psi_m \right\Vert
+\left\Vert (a^*( U_{R}f^+))_{\theta,P}\Psi_m \right\Vert \leq 
\left\Vert  U_{R}f^{+} \right\Vert^2   \left\Vert (N+1)^{1/2 }\Psi_m 
\right\Vert^2
+\\ & \left\Vert  U_{R}f^{-}\right\Vert^2   \left\Vert (N+1)^{1/2 
}\Psi_m\right\Vert^2
=  \left\Vert f^{+} \right\Vert^2   \left\Vert  (N+1)^{1/2 } \Psi_m 
\right\Vert^2
+ \left\Vert f^{-}\right\Vert^2   \left\Vert  (N+1)^{1/2 } 
\Psi_m\right\Vert^2,
\end{align*}
where in the last lines we used the Cauchy-Schwarz inequality, the 
bounds given in \cite{GL1} and the unitarity of  $U_{R}$.
\end{proof}
\section{Properties of the deformed quantum field}
In the following section we prove the Wightman properties of the 
deformed field. The Wightman axioms of covariance and locality are not 
satisfied, but are replaced by wedge covariance and wedge-locality.
The relation between the fields defined on a deformed spacetime and 
fields defined on the wedge is given by the the constructed map in 
\cite{BLS,GL1}. To use this map we give the transformation property of 
the deformed quantum field $\phi_{\theta}$ under Lorentz transformations 
and thus relate the skew-symmetric matrices to wedges.  Furthermore, we 
prove that the field obtained by the construction is a wedge 
Lorentz-covariant and wedge-local quantum field.
\subsection{Wightman properties of the deformed QF}
In this section we prove that the deformed field $\phi_{\theta,K}$ 
satisfies the Wightman properties with the exception of covariance and 
locality.
\begin{proposition}\label{prop1}
  Let $\theta$ be a real skew-symmetric matrix w.r.t. the Minkowski 
scalar-product on $\mathbb{R}^{d}$ and $f \in \mathscr{S}(\mathbb{R}^d)$.
\begin{description}
     \item[a)] The dense subspace $\mathcal{D}$ of vectors of finite 
particle number is contained in the domain 
$\mathcal{D}^{\theta,K}=\{\Psi\in \mathscr{H}|  \left\Vert 
\phi_{\theta,K}(f)\Psi\right\Vert^2 < \infty \}$ of any 
$\phi_{\theta,K}(f)$. Moreover, 
$\phi_{\theta,K}(f)\mathcal{D}\subset\mathcal{D}$ and 
$\phi_{\theta,K}(f)\Omega=\phi(f)\Omega$.
     \item[b)] For scalar fields deformed via warped convolutions and 
$\Psi\in\mathcal{D}$,
\begin{equation}
f\longmapsto\phi_{\theta}(f)\Psi
\end{equation}is a vector valued tempered distribution.
     \item[c)]

For $\Psi\in\mathcal{D}$  and $\phi_{\theta,K}(f)$  the following holds
\begin{equation}
\phi_{\theta,K}(f)^{*}\Psi=\phi_{\theta,K}(\overline{f})\Psi.
\end{equation}
For real $f\in\mathscr{S}(\mathbb{R}^{d})$, the deformed field 
$\phi_{\theta}(f)$ is essentially self-adjoint on
$\mathcal{D}$.
     \item[d)] The Reeh-Schlieder property holds: Given an  open set of 
spacetime $\mathcal{O}\subset \mathbb{R}^d$ then
\begin{equation}
\mathcal{D}_{\theta}(\mathcal{O}):=  
span\{\phi_{\theta}(f_1)\dots\phi_{\theta}(f_m)\Omega: m\in \mathbb{N}, 
f_1\dots f_m\in \mathscr{S}(\mathcal{O})\}
\end{equation}
is dense in $\mathscr{H}$.
  \end{description}
\end{proposition}
\begin{proof} 
a) The fact that $\mathcal{D}\subset \mathcal{D}^{\theta}$,  follows 
directly from lemma \ref{ldk2} because the deformed scalar field 
satisfies the well known bounds of the free field. The fact that the 
deformed field acting on the vacuum is the same as the free field acting 
on $\Omega$, can be easily shown due to the property of the unitary 
operators $U(b)\Omega=\Omega.$
\\\\
b) By using lemma \ref{ldk2} one can see that the right hand side 
depends continuously on the function $f$, hence the temperateness of 
$f\longmapsto\phi_{\theta,K}(f)\Psi$,  $\Psi\in\mathcal{D}$ follows.\\\\
c) The hermiticity of the deformed field  $\phi_{\theta,K}(f)^{*}$ is 
proven in the following
   \begin{align*}
  \phi_{\theta,K}(f)^{*}\Psi&= (2\pi)^{-d}\left(\iint  d^{d}y d^{d}k 
e^{-iy_{\mu}k^{\mu}}\alpha_{\theta y}(\phi(f))U(k)\right)^{*}\Psi 
\\&=(2\pi)^{-d}
  \iint  d^{d}y d^{d}k e^{ iy_{\mu}k^{\mu}}U(-k)\alpha_{\theta 
y}(\phi(f))^* \Psi
  \\&=(2\pi)^{-d}
  \iint  d^{d}y d^{d}k e^{ iy_{\mu}k^{\mu}}\alpha_{\theta 
y}(\phi(\overline{f }))U(-k)\Psi
  =
\phi_{\theta,K}(\overline{f })\Psi, \qquad\qquad    \Psi\in\mathcal{D}.
\end{align*}
In the last lines we performed a variable substitution $(y\rightarrow 
y+\theta^{-1}k)$ and $(k\rightarrow -k)$.  The essential 
self-adjointness of the deformed field for real $f$ can be shown along 
the same lines as  in \cite{BR}.
\\\\
d) For the proof of the Reeh-Schlieder property we will make use of the 
unitary equivalence (\ref{puk}).
  First note that the spectral properties of the representation of the 
special conformal transformations  $U(y)$ are the same as for the 
representation of translations. This leads to the application of the 
standard Reeh-Schlieder argument \cite{SW} which states that that 
$\mathcal{D}_{\theta}(\mathcal{O})$ is dense in $\mathscr{H}$ if and 
only if $\mathcal{D}_{\theta}( {\mathbb{R}^d})$ is dense in 
$\mathscr{H}$.  We choose the functions  $f_1,\dots,f_m \in \mathscr 
{S}(\mathbb{R}^{d})$ such that the Fourier transforms of the functions 
do not intersect the past light cone and therefore the domain 
$\mathcal{D}_{\theta}({\mathbb{R}^d})$ consists of the following vectors
   \begin{align*} &
   \Gamma( U_{R})\phi_{\theta,K}(f_1)\dots \phi_{\theta,K}(f_m)\Omega=
\Gamma( U_{R})a_{\theta,K}^*(f^{+}_1)\dots a_{\theta,K}^*(f^{+}_m)\Omega 
\\&=
\Gamma( U_{R})\Gamma( U_{R})(\Gamma( U_{R})a^*(f^{+}_1)\Gamma( 
U_{R}))_{\theta,P} \dots (\Gamma( U_{R})a^*(f^{+}_m)\Gamma( 
U_{R}))_{\theta,P}\Gamma( U_{R})\Omega\\&=
  a_{\theta,P}^*( U_{R}f^{+}_1)  \dots  a_{\theta,P}^*( U_{R}f^{+}_m) 
\Omega=
  \sqrt{m!}P_{m}(S_m( U_{R}f^{+}_1\otimes \dots \otimes  U_{R} f^{+}_m)),
\end{align*}
where $P_m$ denotes the orthogonal projection from 
$\mathscr{H}_1^{\otimes m}$ onto its totally symmetric subspace 
$\mathscr{H}^{+}_m$, and $S_m \in \mathscr{B}(\mathscr{H}_1^{\otimes 
m})$ is the multiplication operator given as
    \begin{align*} S_m(p_1,\dots,p_m)=\prod \limits_{1\leq l < k \leq m} 
e^{{i}p_l \theta p_k }.
\end{align*}
Since the operator $\Gamma( U_{R})$ is a unitary operator the functions 
$ U_{R}f^{+}_k $  for $f^{+}_k \in \mathscr {S}(\mathbb{R}^{d})$  will 
give rise to dense sets of functions in $\mathscr{H}_1$. Following the 
same arguments as in \cite{GL1} the density of $\mathcal{D}_{\theta}( 
{\mathbb{R}^d})$ in $\mathscr{H}$ follows. Note that we proved the 
density for vectors $\Gamma( U_{R})\phi_{\theta,K}(f_1)\dots 
\phi_{\theta,K}(f_m)\Omega$ and not for the vectors without $\Gamma( 
U_{R})$ as stated in the proposition. We use the unitarity of $\Gamma( 
U_{R})$ to argue that vectors dense in $\mathscr{H}$ stay dense after 
the application of a unitary operator.
\end{proof}
\subsection{Wedge-covariant fields}
The authors in \cite{GL1} constructed a map $Q:W \mapsto Q(W)$ from a 
set $\mathcal{W}_{0}:=\mathscr{L}^{\uparrow}_{+}W_{1}$ of wedges, where 
$W_{1}:=\{x\in \mathbb {R}^d: x_1 > |x_0| \}$  to a set 
$\mathcal{Q}_{0}\subset \mathbb{R}^{-}_{d\times d}$ of skew-symmetric 
matrices. In the next step they considered the corresponding fields 
$\phi_{W}(x):=\phi(Q(W),x)$. The meaning of the correspondence is that 
the field $ \phi(Q(W),x)$ is a scalar field living on a NC spacetime 
which can be equivalently realised as a field defined on the wedge. 
\newline\newline
To show the covariance properties of the deformed quantum fields we use 
the homomorphism $ Q(W)$  to map the deformed scalar fields to quantum 
fields defined on a wedge. Let us first define the following map.

\begin{definition}Let  $\theta$ be a real skew-symmetric matrix on 
$\mathbb{R}^d$ then the map $\gamma_{\Lambda}(\theta)$
is defined as follows
\begin{align}\label{hm}
  \gamma_{\Lambda}(\theta):=
\left\{
\begin{array} {cc}
\Lambda\theta\Lambda^T, \qquad &\Lambda\in\mathcal{L}^{\uparrow}, \\ 
-\Lambda\theta\Lambda^T,\qquad &\Lambda\in\mathcal{L}^{\downarrow} .
\end{array} \right.
\end{align}
\end{definition}
\begin{definition}
$\theta$ is called an admissible matrix if the realisation of the 
homomorphism $Q(\Lambda W)$ defined by the map 
$\gamma_{\Lambda}(\theta)$   is a well defined mapping. This is the case 
iff $\theta$ has in n dimensions the following form
\begin{align}
\begin{pmatrix} 0 & \lambda & 0 & \cdots & 0 \\ \lambda& 0 & 0 & \cdots 
& 0  \\ 0 & 0 &0 &\cdots&0
\\ \vdots & \vdots &\vdots & \ddots& \vdots\\ 0 & 0 &0 & \cdots&0
  \end{pmatrix},\qquad \lambda\geq0.
\end{align}
For the physical most interesting case of 4 dimensions the 
skew-symmetric matrix $\theta$ has the more general form
\begin{align}\qquad\quad
\begin{pmatrix} 0 & \lambda & 0 & 0 \\ \lambda& 0 & 0 & 0  \\ 0 & 0 &0 &\eta
  \\ 0 & 0 &-\eta &0
  \end{pmatrix},\qquad \lambda\geq0, \eta \in \mathbb{R}.
\end{align}

\end{definition}
Before we use the map from the set of skew-symmetric matrices to the 
wedges we state the following lemma about the transformation properties 
of the deformed field.
\begin{lemma} \label{lemma2}
The transformation of the deformed particle annihilation and creation 
operator $a_{\theta,K}(\mathbf{p}),a_{\theta,K}^{*}(\mathbf{p})$, for 
$\mathbf{p}\in\partial V_{+}$ and $\theta$ being admissible, under the 
adjoint action $V(0,\Lambda)$ of the Lorentz group, 
$\Lambda\in\mathcal{L}$, is the following
\begin{align}
V(0,\Lambda)a_{\theta,K}(\mathbf{ p}) V(0,\Lambda)^{-1}=
a_{\gamma_{\Lambda}(\theta),K}(\pm \Lambda \mathbf{ p} ),\\
V(0,\Lambda)a^{*}_{\theta,K}( \mathbf{ p}) V(0,\Lambda)^{-1}=
a_{\gamma_{\Lambda}(\theta),K}^{*}(\pm \Lambda \mathbf{ p}  ),
\end{align}
where the first sign is for  $\Lambda\in\mathcal{L}^{\uparrow}$ and the 
second sign is for  $\Lambda\in\mathcal{L}^{\downarrow}$. Hence the 
deformed field $\phi_{\theta,K}( x)$ transforms
\begin{equation}
V(0,\Lambda)\phi_{\theta,K} ( x)V(0,\Lambda)^{-1}=
\phi_{\gamma_{\Lambda}(\theta) ,K}( \Lambda x).
\end{equation}
\end{lemma}
\begin{proof}
The proof is done along the lines of \cite{BLS}. $V(0,\Lambda)$ is a 
unitary operator for $\Lambda\in\mathcal{L}^{\uparrow}$ and an 
antiunitary operator for $\Lambda\in\mathcal{L}^{\downarrow}$. Due to 
the commutator relation of the special conformal operator and the 
generator of the Lorentz transformations one obtains
\begin{equation}
V(0,\Lambda)U(x)V(0,\Lambda)^{-1}=U(\Lambda x), \qquad x \in \mathbb{R}^d.
  \end{equation}
Therefore, the deformed scalar field $\phi_{\theta,K}$ transforms under 
the adjoint action of the Lorentz group in the following way
\begin{align*}
V(0,\Lambda)
\phi_{\theta,K} ( x)V(0,\Lambda)^{-1}&= (2\pi)^{-d}V(0,\Lambda)\iint  
d^{d}y d^{d}k e^{ -iy_{\mu}k^{\mu}}\alpha_{\theta y}(\phi(x))U( 
k)V(0,\Lambda)^{-1}\\&=(2\pi)^{-d}
\iint  d^{d}y d^{d}k e^{ - i\sigma y_{\mu}k^{\mu}} \alpha_{\Lambda 
\theta y}( V(0,\Lambda)\phi( x)V(0,\Lambda)^{-1})U( \Lambda k)\\ 
&= (2\pi)^{-d}
\iint  d^{d}y d^{d}k e^{ 
-iy_{\mu}k^{\mu}}\alpha_{\gamma_{\Lambda}(\theta)    y}(  \phi(
\Lambda x) ) U(   k)\\&=\phi_{\gamma_{\Lambda}(\theta),K} (\Lambda x),
\end{align*}
  where $\sigma$ is $+1$ if V is unitary and $-1$ if V is antiunitary. 
Moreover in the last lines we performed the integration variable 
substitutions $(y,k)\rightarrow(\sigma\Lambda^Ty, \Lambda^{-1}k)$. 
\end{proof} 
In the next step we use the homomorphism (\ref{hm}) to map the deformed 
field to a field defined on a wedge. Furthermore we show that the field 
deformed with the special conformal operator is a wedge-covariant 
quantum field which transforms covariantly under the adjoint action of 
the Lorentz group $V(0,\Lambda)$. For this purpose let us first 
introduce the notion of a  wedge-covariant quantum field, \cite{GL1}.

\begin{definition} Let $\phi=\{\phi_W: W\in\mathcal{W}_{0}\}$ denote the family of fields satisfying the domain and continuity assumptions of the Wightman axioms. Then the field  $\phi $ is 
defined to be a wedge Lorentz-covariant quantum field if the following 
condition is satisfied: \begin{itemize}
      \item   For any $W \in\mathcal{W}_{0} $ and 
$f\in\mathscr{S}(\mathbb{R}^d)$ the following holds
\begin{align*}
      V( \Lambda)\phi_W(f)V( \Lambda)^{-1}&=\phi_{\Lambda W}(f\circ( 
\Lambda)^{-1} ),\qquad  \Lambda\in \mathcal{L}^{\uparrow}_{+},
\\
      V( j)\phi_W(f)V( j)^{-1}&=\phi_{ jW}(\overline{f}\circ j)^{-1} .
\end{align*}

   \end{itemize}

\end{definition}
We use the homomorphism $Q:W \mapsto Q(W)$  to define the deformed fields 
as quantum fields defined on the wedge, this is done in the following way

\begin{equation}\label{e1}
\phi_W(f):=\phi(Q(W),f)=\phi_{\theta,K}( f).
\end{equation}

\begin{proposition}  The family of fields $\phi=\{\phi_W: W\in\mathcal{W}_{0}\}$ 
defined by the deformation with the special conformal operators are 
wedge-covariant quantum fields on the Bosonic Fock space, w.r.t. the 
unitary representation $V$ of the Lorentz group.
  \end{proposition}
\begin{proof}
 Following lemma (\ref{lemma2}), the 
deformed field $\phi_{\theta,K} (x)$ transforms under the adjoint action  
$V$ of the Lorentz group in the following way
\begin{equation}V(0,\Lambda)\phi_{W} (x)V(0,\Lambda)^{-1}=
V(0,\Lambda)\phi_{\theta,K} (x)V(0,\Lambda)^{-1}=
\phi_{\gamma_{\Lambda}(\theta) ,K}( \Lambda x)=
\phi_{\Lambda W,K}( \Lambda x),
\end{equation}
where in the last lines we applied the map $Q(\Lambda 
W)=\gamma_{\Lambda}(Q(W))=\gamma_{\Lambda}(\theta)$. Therefore, one 
obtains the wedge-covariance property of the scalar field under the 
Lorentz group.
\end{proof}
A few comments are in order. The covariance property is given in the 4 
dimensional case as well. As explained in Sec. 3, a unitary 
representation for the whole conformal group  does not exist due to the 
absolute value of the scale factor.   We will show in the next section 
that the scale factor is positive for a field localised in the wedge and 
therefore one has a unitary representation of the whole conformal group.
\subsection{Wedge-local fields}
In this section we will show that the wedge-covariant quantum field 
defined in the last section, is a \textit{wedge-local} field. We first 
define the notion of the wedge-local field.

\begin{definition} The fields $\phi=\{\phi_W: W\in\mathcal{W}_{0}\}$ are 
said to be wedge-local if the following commutator relation is satisfied
\begin{equation}
[\phi_{W_1}(f),\phi_{-W_1}(g)]\Psi=0,\qquad \Psi\in\mathcal{D},
\end{equation}
for all $f,g \in C_{0}^{\infty}(\mathbb{R}^d)$ with $\mathrm{supp}$ 
$f\subset W_1 $ and $ \mathrm{supp} $ $g\subset -W_1 $.
\end{definition}
To show that the fields defined in the last section are wedge-local, we use the following proposition (2.10, \cite{BLS}) and lemma \ref{lkw}.

\begin{proposition} \label{p1}Let the scalar fields $\phi(f)$, $\phi(g)$ 
be such that $[\alpha_{\theta v}(\phi(f)),\alpha_{-\theta 
u}(\phi(g))]=0$ for all $v,u$ $\in$ spU and for $f,g \in 
C_{0}^{\infty}(\mathbb{R}^d)$. Then
\begin{equation}\label{wlf}
[ \phi_{\theta,K}(f) ,\phi_{-\theta,K}(g)]\Psi=0,\qquad \Psi\in\mathcal{D} .
\end{equation}
\end{proposition}
\begin{lemma} \label{lkw}The special conformal transformations $U_{\theta v}$, with $v$ $\in$ spU and 
$\theta$ being admissible, map the right wedge into the right wedge  
$U_{\theta v}({W}_{1})\subset {W}_{1}$. Furthermore, the special conformal transformations $U_{-\theta u}$, with $u$ $\in$ spU and $\theta$ being admissible, map the left wedge into the left wedge  
$U_{-\theta v}(-{W}_{1})\subset -{W}_{1}$.
\end{lemma}
\begin{proof} 
We first prove for  $x^{\mu}\in W_1 $,  $v$ $\in$ spU, $\theta$ 
being admissible and $\kappa>0$, that the vector $x'^{\mu}:=x^{\mu}+\kappa (\theta 
v)^{\mu} \in W_1$.
\begin{align*}
x'^1 &>|x'^0 |\\
x^{1}+\kappa \lambda v^{0}&>|x^{0}+\kappa \lambda v^{1} |.
\end{align*}
The right hand side is obviously greater than zero and therefore we 
square both sides and obtain
\begin{align*}
\kappa^2 \lambda^2 (v_0^{2}-v_1^{2})-(x_0^2-x_1^2)-2\kappa \lambda (   
v^1 x^{0}-  v^{0} x^1)
&>0.
\end{align*}
Due to the fact that the sum of the first two terms is greater than 
zero, we are only left with proving that the following inequality
\begin{align*}
   \lambda v^1 x^{0}-  \lambda v^{0} x^1\leq 0
\end{align*}
is satisfied.
Equality only holds if $v_{0}=0$ or $\lambda$ is zero. So if 
$v_0,\lambda\neq0$ we have to show the following
\begin{align} \label{i1}
     x^1>\frac{v^1}{v^
     0} x^{0}.
\end{align}
(\ref{i1}) is satisfied, because the stronger inequality
\begin{align*}
     x_1>\underbrace{|\frac{v_1}{v_
     0}|}_{0<\cdots<1}| x_{0}|
\end{align*}
holds. By using the vector $x'^{\mu}$ we now can easily prove that 
$x^{\mu}_{\theta v}\in W_1$.  To show that $x^{\mu}_{\theta v}\in W_1$ 
the following inequality must be satisfied.
\begin{align}
x^1_{\theta v}&>|x^0_{\theta v}|\end{align}  \begin{align}\label{u1}
({x^{1}-{(\theta v)}^{1}x^2})/({1-2(\theta v)\cdot x+(\theta v)^2x^2})&>|
({x^{0}-{(\theta v)}^{0}x^2})/({1-2(\theta v)\cdot x+(\theta v)^2x^2})|.
\end{align}
Positivity of the denominator can be seen by taking the vector 
$x'^{\mu}$ as defined above and setting $\kappa =-x^2>0$. From $x'^2<0$ 
we obtain
\begin{align*}
  x'^2=(x^{\mu}-x^2 (\theta v)^{\mu})(x_{\mu}-x^2 (\theta v)_{\mu})=
\underbrace{x^2}_{<0}\left(1-2x_{\mu}(\theta v)^{\mu}+x^2 (\theta 
v)_{\mu} \right)<0.
\end{align*}
 From the inequality it follows that the the denominator in (\ref{u1}) 
is positive and therefore one is left with proving
\begin{align}
({x^{1}-{(\theta v)}^{1}x^2})&>|
({x^{0}-{(\theta v)}^{0}x^2})|.
\end{align}
By choosing $\kappa =-x^2$ this is exactly the inequality for the vector 
$x'^{\mu} \in W_1$. Therefore, the special conformal transformed 
coordinate is still in the right wedge. The proof that the special 
conformal transformations map the left wedge into the left wedge is 
analogous. \end{proof} 
\begin{proposition} \label{propt}
For $n=4l+1$, where $l\in \mathbb{N}_0$ the family of fields $\phi=\{\phi_W: 
W\in\mathcal{W}_{0}\}$ 
are wedge-local fields on the Bosonic Fockspace $\mathscr{H}^+ $.
\end{proposition}
\begin{proof} We first prove that the expression
$[\alpha_{\theta v}(\phi(f)),\alpha_{-\theta u}(\phi(g))]$ vanishes for 
all $v,u$ $\in$ spU and for $f \in C_{0}^{\infty}(W_1)$, $g \in 
C_{0}^{\infty}(-W_1)$. By using proposition \ref{p1} it then follows, that the commutator
$[ \phi_{W_1}(f) ,\phi_{-W_1}(g)]$ vanishes.  
\begin{align*}&
  [\alpha_{\theta v}(\phi(f)),\alpha_{-\theta u}(\phi(g))]=(2\pi)^{-2(n+1)}
  \iint d^{n+1}x d^{n+1}y
f(x)g(y)[\alpha_{\theta v}(\phi(x)),\alpha_{-\theta u}(\phi(y))]\\
&= (2\pi)^{-2(n+1)} \iint d^{n+1}x d^{n+1}y
f(x)g(y) {\sigma_{\theta v}(x)}^{ \frac{1-n}{2}} {\sigma_{-\theta 
u}(y)}^{ \frac{1-n}{2}} [\phi(x_{\theta v}), \phi(y_{-\theta u})]  = 0.
\end{align*}
In the last line we applied lemma \ref{lkw} to prove that after the special 
conformal transformation, the support of the field $\phi_{W_1}$ stays in 
the right wedge and the support of the field $\phi_{-W_1}$ stays in the 
left wedge. Therefore, the supports of the fields are space-like separated, hence they
commute.
 \end{proof} 
\begin{lemma} \label{l1}In four dimensions a unitary representation for the whole conformal 
group, which gives the correct transformation law (\ref{trafo1}), exists for the fields $\phi_{\theta,K}(f)$ with 
$f\in C_{0}^{\infty}(W_1)$. The same holds for the field $\phi_{-\theta,K}(g)$ with 
$g\in C_{0}^{\infty}(-W_1)$.
\end{lemma}
\begin{proof} 
The problem with the absence of a unitary representation for
the whole conformal group that gives the correct transformation law 
(\ref{trafo1}) is due to
the absolute value of the scale factor $\sigma_b(x)$. Nevertheless, we 
showed in lemma \ref{lkw}
that the scale factor for a field localised in the right wedge is 
positive. The positivity of the
scale factor in turn  means that a unitary representation for the whole 
conformal group in four
dimensions exists, \cite{SV}.
\begin{align*}
  \phi_{W_1}(f)\Psi=\phi_{\theta,K}(f)\Psi=(2\pi)^{-4}\int d^4x f(x) 
\iint d^4vd^4u e^{-iuv}\alpha_{\theta v}(\phi(f))U(u)\Psi
  \\ =(2\pi)^{-4}\int d^4x f(x)\iint d^4v d^4u e^{-iuv} {\sigma_{\theta 
v}(x)}^{-1}\phi(x_{\theta v})
  U(u)\Psi,\qquad \Psi\in\mathcal{D}.
\end{align*}
For a quantum field defined on the left wedge the proof is done 
analogously.
\end{proof} 
\begin{proposition} For  $n=3$, the fields $\phi=\{\phi_W: 
W\in\mathcal{W}_{0}\}$ 
are wedge-local fields on the Bosonic Fockspace $\mathscr{H}^+ $.
\end{proposition}
\begin{proof} 
Due to the existence of a unitary representation shown in lemma \ref{l1} 
the deformed field can be defined for $n=3$. Furthermore, by applying 
\ref{p1} one shows that the expression $[\alpha_{\theta 
v}(\phi(f)),\alpha_{-\theta u}(\phi(g))]$ vanishes for all $v,u$ $\in$ 
spU and for $f \in C_{0}^{\infty}(W_1)$, $g \in C_{0}^{\infty}(-W_1)$.
\begin{align*}&
  [\alpha_{\theta v}(\phi(f)),\alpha_{-\theta u}(\phi(g))]=(2\pi)^{-8}
  \iint d^{4}x d^{4}y
f(x)g(y)[\alpha_{\theta v}(\phi(x)),\alpha_{-\theta u}(\phi(y))]\\
&= (2\pi)^{-8} \iint d^{4}x d^{4}y
f(x)g(y) {\sigma_{\theta v}(x)}^{ -1} {\sigma_{-\theta u}(y)}^{ -1} 
[\phi(x_{\theta v}), \phi(y_{-\theta u})]  = 0,
\end{align*}
Analogous to the proof of proposition (\ref{propt}) we use   lemma \ref{lkw} in the last line.\end{proof}
This is a very interesting result. The deformed case improves the representations such that one does not have to deal with representations of the covering of the conformal group.
\section{ NC Spacetime from special conformal operators}

\subsection{Moyal-Weyl}
Within the framework of warped convolutions \cite{BLS}, the authors 
defined a deformed associative product in the following way.

\begin{definition}\label{dp}The associative deformed product $\times_{\theta}$ of 
$A,B$ is defined as
\begin{equation}
A\times_{\theta}B=(2\pi)^{-d}\iint d^dv d^du e^{-ivu}\alpha_{\theta 
v}(A)\alpha_{u}(B).
\end{equation}
Furthermore, the deformed commutator $[A\stackrel{\times_{\theta}}{,}B]$ of A, B is defined   in the following way 

\begin{equation}\label{dc}
[A\stackrel{\times_{\theta}}{,} B] :=A\times_{\theta}B-B\times_{\theta}A.
\end{equation}

\end{definition}The deformed product can be used to calculate the commutator of the coordinates. In the case of deformation  with the momentum operator $P_{\mu}$ one obtains the following lemma.
\begin{lemma}\label{lwp}
Let the deformed product \ref{dp} be  defined by the generator  of translations $P_{\mu}$. Then the deformed commutator \ref{dc}  of the coordinates gives the Moyal-Weyl plane  
\begin{equation}
[x_{\mu}\stackrel{\times_{\theta}}{,} x_{\nu}] :=x_{\mu}\times_{\theta}x_{\nu}-x_{\nu}\times_{\theta}x_{\mu}=-2i\theta_{\mu\nu}.
\end{equation}
\end{lemma}
\begin{proof}We first calculate the deformed product of the coordinates using definition \ref{dp}.
\begin{align*} 
x_{\mu}\times_{\theta}x_{\nu}&=(2\pi)^{-d}\iint d^dv d^du 
e^{-ivu}\alpha_{\theta v}(x_{\mu})\alpha_{u}(x_{\nu})\\&
=(2\pi)^{-d}\iint d^dv d^du e^{-ivu} (x_{\mu}+{(\theta v)}_{\mu}) 
(x_{\nu}+{u}_{\nu}) =x_{\mu}x_{\nu}-i\theta_{\mu\nu}
\end{align*}
In the last lines, we applied the adjoint action of the momentum operator 
$P_{\mu}$ on the coordinates, which induces a translation. The next 
step consists in calculating the deformed commutator of the coordinates.
Due to the skew-symmetry of the deformation matrix $\theta$, one 
obtains for the deformed commutator the Moyal-Weyl plane.
\end{proof}
This result is not surprising. As already mentioned, in \cite{GL1} a quantum field was defined on the Moyal-Weyl plane, which also can obtained  by using the momentum operator for deformation via warped convolutions. 
Therefore, it is only natural that the Moyal-Weyl plane appears for the deformed commutator of the coordinates. In the next section we will calculate the commutator of the coordinates 
by using the deformed product induced by the special conformal operators.
\subsection{Nonconstant noncommutative spacetime}
The main idea in this work is to use the special conformal operator to 
deform the free quantum field. We further proved that the deformed field 
satisfies some weakened covariance and locality properties. Now a 
natural question arises. What is the noncommutative spacetime that we 
obtain from the deformation with the special conformal operator? This 
question can be answered by calculating the deformed commutator of the 
coordinates.
\begin{equation}\label{comm}
[x_{\mu}\stackrel{\times_{\theta}}{,} x_{\nu}] =(2\pi)^{-d}\iint d^dv 
d^du e^{-ivu}\left( \alpha_{\theta 
v}(x_{\mu})\alpha_{u}(x_{\nu})-\alpha_{\theta v}(x_{\nu})\alpha_{u}(x_{\mu})
\right).
\end{equation}
To calculate the term $\alpha_{\theta v}(x_{\mu})$ we insert the 
generator $K_{\mu}$ as a differential operator defined in (\cite{DMS}).
\begin{equation}
\alpha_{\theta v}(x_{\mu})=\mathrm{exp}\left({\left(\theta 
v\right)^{\sigma}\left(2x_{\sigma}x^{\lambda}\frac{\partial}{\partial 
x^{\lambda}}-
x^2\frac{\partial}{\partial x^{\sigma}}
\right)}\right)x_{\mu}=:\mathrm{exp}\left( (\theta v)^{\sigma}K_{\sigma}(x)
\right)x_{\mu}
\end{equation}
We could use the transformation of the coordinates under the special 
conformal generators, but in that case we would not be able to solve the 
integral. The ansatz we follow in this work is to solve the integral, 
order by order. This will be done by preforming a Taylor expansion of 
the exponentials.
\begin{lemma}
 Let the deformed product \ref{dp} be  defined by the generator  of special conformal transformations $K_{\mu}$. Then the deformed commutator (\ref{comm}), up to third order in $\theta$ is given as follows
\begin{align*}
[x_{\mu}\stackrel{\times_{\theta}}{,} x_{\nu}]&=
-2i\theta_{\mu\nu}x^4
-4i\left((\theta x)_{\mu}x_{\nu}-(\theta x)_{\nu}x_{\mu} 
\right)x^{2}.
\end{align*}
\end{lemma}
\begin{proof}
The deformed commutator gives the following
\begin{align*}
&[x_{\mu}\stackrel{\times_{\theta}}{,} x_{\nu}]=(2\pi)^{-d}\iint d^dv 
d^du e^{-ivu}\left(\alpha_{\theta 
v}(x_{\mu})\alpha_{u}(x_{\nu})-\mu\leftrightarrow\nu\right) \\&
=(2\pi)^{-d}\iint d^dv d^du e^{-ivu}\left(\sum_{k=0}^{\infty}
\frac{1}{k!}\underbrace{(\theta v)^{\sigma}K_{\sigma}(x)\cdots (\theta 
v)^{\rho}K_{\rho}(x)}_{k}x_{\mu}\right)
\end{align*}
\begin{align*}&
\times\left(\sum_{l=0}^{\infty}
\frac{1}{l!}\underbrace{(u)^{\lambda}K_{\lambda}(x)\cdots 
(u)^{\tau}K_{\tau}(x)}_{l}x_{\nu}
-\mu\leftrightarrow\nu\right).
\end{align*}
There are two properties for the series that can be easily seen. First, 
the different orders between $\theta v$
and $u$ do not mix. The only terms which are not equal to zero are the 
terms of equal order. The vanishing of unequal orders between $\theta v$ 
and $u$ will be shown in the following calculation.

\begin{align*}
&\iint d^dv d^du e^{-ivu}\sum_{k=0}^{\infty}
\sum_{l=0}^{\infty}\frac{1}{k!l!}
\underbrace{(\theta v)^{\sigma}\cdots (\theta v)^{\rho}}_{k}
\underbrace{u^{\lambda}\cdots u^{\tau}}_{l}\left(
\underbrace{K_{\sigma}(x)\cdots K_{\rho}(x)}_{k}x_{\mu} \right)\\&
\times\left(
\underbrace{K_{\lambda}(x)\cdots K_{\tau}(x)}_{l}x_{\nu} 
\right)\end{align*}
\begin{align*}
&
=
\iint d^dv d^du \sum_{k=0}^{\infty}
\sum_{l=0}^{\infty}\frac{(-i)^k}{k!l!}
\underbrace{\theta^{\sigma}_{\kappa}\cdots \theta^{\rho}_{\gamma}
}_{k}
  \left( \underbrace{\frac{\partial}{\partial u_{\kappa}}\cdots 
\frac{\partial}{\partial u_{\gamma}}}_{k}
e^{-ivu}\right) \underbrace{u^{\lambda}\cdots u^{\tau}}_{l}
\\&\times
\left(
\underbrace{K_{\sigma}(x)\cdots K_{\rho}(x)}_{k}x_{\mu} \right)\left(
\underbrace{K_{\lambda}(x)\cdots K_{\tau}(x)}_{l}x_{\nu} 
\right)\end{align*} \begin{align*} &=
\iint d^dv d^du e^{-ivu}\sum_{k=0}^{\infty}
\sum_{l=0}^{\infty}\frac{i^k}{k!l!}
\underbrace{\theta^{\sigma}_{\kappa}\cdots \theta^{\rho}_{\gamma}
}_{k} \left( \left( \underbrace{\frac{\partial}{\partial 
u_{\kappa}}\cdots \frac{\partial}{\partial u_{\gamma}}}_{k}
\right) \underbrace{u^{\lambda}\cdots u^{\tau}}_{l}\right)
\\&\times
\left(
\underbrace{K_{\sigma}(x)\cdots K_{\rho}(x)}_{k}x_{\mu} \right)\left(
\underbrace{K_{\lambda}(x)\cdots K_{\tau}(x)}_{l}x_{\nu} 
\right)\end{align*} \begin{align*}&
=\sum_{k=0}^{\infty}\frac{(-i)^k}{k!}
\underbrace{\theta^{\sigma\lambda}\cdots \theta^{\rho\tau}}_{k}\left(
\underbrace{K_{\sigma}(x)\cdots K_{\rho}(x)}_{k}x_{\mu} \right)\left(
\underbrace{K_{\lambda}(x)\cdots K_{\tau}(x)}_{k}x_{\nu} \right)
\end{align*}
In the third line we performed a partial integration. The expression 
vanishes in the case $k>l$, because the differentials annihilate the 
polynomial in $u$. It also vanishes if $k<l$ because  nonvanishing 
polynomials in $u$ stay and the integral sets the polynomials zero. 
Furthermore, by using the symmetry of the $x$-dependent differential 
operators $K$, one solves the integral. It is important to note that the 
result of the deformed product between the coordinates, is exactly the 
same result one would encounter by using twist-deformation with the 
special conformal operators, \cite{AJ}.
\\ \\
The second observation is that polynomials in $u,v$ that are even vanish 
due to the antisymmetry of the commutator. This is shown in the following.
\begin{align*}&
\iint d^dv d^du e^{-ivu}\left(
\underbrace{(\theta v)^{\sigma}K_{\sigma}(x)\cdots (\theta 
v)^{\rho}K_{\rho}(x)}_{2m}x_{\mu}
\underbrace{(u)^{\lambda}K_{\lambda}(x)\cdots 
(u)^{\tau}K_{\tau}(x)}_{2m}x_{\nu} \right)
\\ -&\iint d^dv d^du e^{-ivu}
\left(
\underbrace{(\theta v)^{\sigma}K_{\sigma}(x)\cdots (\theta 
v)^{\rho}K_{\rho}(x)}_{2m}x_{\nu}
\underbrace{(u)^{\lambda}K_{\lambda}(x)\cdots 
(u)^{\tau}K_{\tau}(x)}_{2m}x_{\mu} \right)
\end{align*}
Where m is a natural number. In the second integral we preform the 
integration variable substitution
$(v,u)\rightarrow (\theta^{-1}u,\theta v) $ and obtain
\begin{align*}&
\iint d^dv d^du e^{-ivu}\left(
\underbrace{(\theta v)^{\sigma}K_{\sigma}(x)\cdots (\theta 
v)^{\rho}K_{\rho}(x)}_{2m}x_{\mu}
\underbrace{(u)^{\lambda}K_{\lambda}(x)\cdots 
(u)^{\tau}K_{\tau}(x)}_{2m}x_{\nu} \right)
\\ -&\iint d^dv d^du e^{ivu}
\left(
\underbrace{(\theta v)^{\sigma}K_{\sigma}(x)\cdots (\theta 
v)^{\rho}K_{\rho}(x)}_{2m}x_{\mu}
\underbrace{(u)^{\lambda}K_{\lambda}(x)\cdots 
(u)^{\tau}K_{\tau}(x)}_{2m}x_{\nu} \right).
\end{align*}
After preforming the integration variable substitution $u\rightarrow -u$ 
we obtain
\begin{align*}&
\iint d^dv d^du e^{-ivu}\left(
\underbrace{(\theta v)^{\sigma}K_{\sigma}(x)\cdots (\theta 
v)^{\rho}K_{\rho}(x)}_{2m}x_{\mu}
\underbrace{(u)^{\lambda}K_{\lambda}(x)\cdots 
(u)^{\tau}K_{\tau}(x)}_{2m}x_{\nu} \right)
\\ -&(-1)^{2m}\iint d^dv d^du e^{-ivu}
\left(
\underbrace{(\theta v)^{\sigma}K_{\sigma}(x)\cdots (\theta 
v)^{\rho}K_{\rho}(x)}_{2m}x_{\mu}
\underbrace{(u)^{\lambda}K_{\lambda}(x)\cdots 
(u)^{\tau}K_{\tau}(x)}_{2m}x_{\nu} \right)\\&=0.
\end{align*}
Therefore, the only terms that do not vanish are those of equal odd 
order in $v$ and $u$.
In the following we calculate the noncommutativity of the coordinates 
up to the second order and obtain  
\begin{align*}\label{ncst}
[x_{\mu}\stackrel{\times_{\theta}}{,} x_{\nu}]&
=(2\pi)^{-d}\iint d^dv d^du e^{-ivu}\left(
     (\theta v)^{\sigma}(
2x_{\sigma}x_{\mu}-
x^2\eta_{\sigma\mu})
     (u)^{\tau}
\left(2x_{\tau}x_{\nu}-
x^2\eta_{\tau\nu}\right) - \mu\leftrightarrow\nu
\right)\\&
=
-2i\theta_{\mu\nu}x^4
-4i\left((\theta x)_{\mu}x_{\nu}-(\theta x)_{\nu}x_{\mu} 
\right)x^{2}+\mathcal{O}(\theta^3).
\end{align*}\end{proof}
The deformed commutator of the coordinates shows that the deformation 
induced by the special conformal operators spans a nonconstant 
noncommutative spacetime. 
 This is very interesting because 
the spacetime that we obtain is a curved noncommutative spacetime and 
the curvature of the noncommutative spacetime is induced by the special 
conformal operators. In the case of using the momentum operator, i.e. 
the generator of translations in Minkowski, for the deformation one 
obtains a flat noncommutative spacetime. The special conformal 
operators  induce a conformal flat spacetime on Minkowski  and 
therefore, one obtains a conformally flat noncommutative spacetime when 
deforming with $K_{\mu}$. Some examples of a nonconstant noncommutative 
spacetime exist in literature, where the  highest order of 
noncommutativity known is the so called quantum space structure, \cite{KS, W1, W2}. The quantum 
space structure has an $x$-polynomial dependence up to second order.

\subsection{Generalisation of the deformation}
The deformation of an operator   by either using the momentum operator 
$P_{\mu}$ or the special conformal operator $K_{\mu}$ can be written in 
a general form. The generalisation can be accomplished by using a linear 
combination of generators of the pseudo-orthogonal group $SO(2,d)$. 
First, we redefine the operators $P^{\mu}$ and $K^{\nu}$ in the 
following way
\begin{align}\tilde{P}^{\mu}:=
\begin{pmatrix} {\lambda'}  P^0  \\ {\lambda'} P^1 \\  {\eta'}  P^2
  \\   {\eta'}  P^3
  \end{pmatrix},
\qquad \tilde{K}^{\mu}:=
\begin{pmatrix} {\lambda}  K^0  \\ {\lambda} K^1 \\  {\eta}  K^2
  \\   {\eta}  K^3 \end{pmatrix},
\end{align} where $\lambda',\lambda\in \mathbb{R}^{+}$ and 
$\eta',\eta\in \mathbb{R}$.
In the next step we redefine the Lorentz generators $J^{4,\mu}$, 
$J^{5,\mu}$, $J^{\pm,\mu}$ and the skew-symmetric matrix $\theta$ as follows
\begin{align}\tilde{J}^{4,\mu}:=
\frac{1}{2}\left(\tilde{P}^{\mu}-\tilde{K}^{\mu}\right),& \qquad 
\tilde{J}^{5,\mu}:=
\frac{1}{2}\left(\tilde{P}^{\mu}+\tilde{K}^{\mu}\right)
\end{align}
\begin{align}
\tilde{J}^{\pm,\mu}:= \tilde{J}^{5,\mu} \pm \tilde{J}^{4,\mu},
\end{align}
\begin{align}\label{th}
\tilde{\theta}=
\begin{pmatrix} 0 & 1 & 0&0 \\ 1 & 0 & 0&0\\0 & 0& 0&1\\0 & 0 & -1&0
  \end{pmatrix}.
\end{align}
\begin{definition}
Let $\tilde{\theta}$ be a real skew-symmetric matrix given in (\ref{th}) 
and let $A\in C^{\infty}$. Then the generalized warped convolutions, 
i.e. the deformation of $A$ denoted as $A^{\pm}_{\theta}$ is defined as 
follows
\begin{equation}
A^{\pm}_{\theta}\Psi:=(2\pi)^{-d}
   \iint  d^{4}y d^{4}ke^{-iy_{\mu}k^{\mu}} U^{\pm}(\theta 
y)AU^{\pm}(\theta y)U^{\pm}(k)\Psi, \qquad \Psi \in \mathcal{D},
\end{equation}
where the unitary operator $U^{\pm}(k)$ is defined as  
$U^{\pm}(k):=\mathrm{exp}\left(i k^{\mu}\tilde{J}^{\pm}_{\mu}\right)$.
\end{definition}
The generalisation of the deformation is interesting because it is 
obtained as a linear combination of generators of $SO(2,4)$. By choosing 
the plus sign, one obtains the Moyal-Weyl case and by choosing the minus 
sign one gets the special conformal model introduced in this work.
\section{Conclusion and outlook}
In this work we deformed a quantum field 
theory with the special conformal operator $K_{\mu}$, using the warped convolutions.  To proceed with the deformation, self-adjointness of the generator $K_{\mu}$ is proven in order to obtain a strongly continuous automorphism of the group $\mathbb{R}^n$. The proof of
self-adjointness was done rigorously in \cite{SV} and is sketched in Sec. 3. Therefore, we were able to 
define the deformation of the scalar field with the special conformal 
operators. We further proved that the deformed quantum field satisfies 
the Wightman axioms, except for the covariance and locality.  \\\\
The homomorphism $Q(\Lambda W)$, defined in \cite{GL1} was used in 
this work to define the map from the deformed field $\phi_{\theta}$ to a 
field defined on the wedge $\phi_{W}$. Furthermore, it was proven that 
the field $\phi_{W}$ transforms as a  wedge-covariant field under the 
adjoint action of the Lorentz group. Wedge-locality for the field 
$\phi_{W}$ was shown in $d=4l+2,l \in\mathbb{N}_0$ dimensions. In 4 
dimensions one usually has a problem with the existence of a unitary 
representation for the whole conformal group. The absence of a unitary 
representation is due to the absolute value of the scale factor induced 
by the special conformal transformations. We circumvented the problem by 
proving positivity of the scale factor. Positivity was proven by using 
the properties of the wedge and the spectrum condition of the special 
conformal operator.
\\\\
The deformed product defined in \cite{BLS} is used to understand the 
noncommutative spacetime being induced by using the special conformal 
operators. The deformed product is used to calculate the commutator of 
the coordinates. We first proved that the formula obtained by solving 
the integral is known in literature as twist deformation. Furthermore, 
we discovered that the noncommutative spacetime obtained in this manner 
is a nonconstant noncommutative spacetime which seems to be a new 
result.  
\\\\
To calculate the S-matrix in the current framework we have to use 
the concept of the temperate polarization-free generators defined in 
\cite{BBS}. The concept can be used for fields deformed with the 
momentum operator, but for the special conformal operator we still have 
to work out some technical subtleties. This will be done in a further work.
\section*{Aknowledgments}
I would like to thank my advisor K. Sibold for providing constant support. Furthermore, I would
like to thank  G. Lechner for very helpful remarks and in particular, S. Alazzawi for technical details concerning 
the Reeh-Schlieder property. I would like to thank S. Pottel for many fruitful discussions and  
T. Ludwig and J. Zschoche for proof reading.  The hospitality of the Max Planck Institute is gratefully acknowledged.

\end{document}